\documentclass[11pt]{article}
\usepackage[margin=1in]{geometry}
\usepackage{complexity}
\usepackage{mathtools,amsmath,amssymb}
\usepackage{xspace}
\RequirePackage[colorlinks=true]{hyperref}
\hypersetup{
  linkcolor=[rgb]{0,0,0.5},
  citecolor=[rgb]{0, 0.5, 0},
  urlcolor=[rgb]{0.5, 0, 0}
}
\usepackage{dsfont}

\usepackage{amsthm}
\usepackage{thmtools,thm-restate}

\numberwithin{equation}{section}
\declaretheoremstyle[bodyfont=\it,qed=\qedsymbol]{noproofstyle}

\declaretheorem[name=Observation,numbered=no]{observation*}

\declaretheorem[numberlike=equation]{theorem}

\declaretheorem[name=Theorem,numbered=no]{theorem*}

\declaretheorem[numberlike=equation]{lemma}
\declaretheorem[name=Lemma,numbered=no]{lemma*}

\declaretheorem[name=Corollary,numbered=no]{corollary*}

\declaretheorem[numberlike=equation]{proposition}
\declaretheorem[name=Proposition,numbered=no]{proposition*}

\declaretheorem[name=Claim,numbered=no]{claim*}

\declaretheorem[name=Conjecture,numbered=no]{conjecture*}

\declaretheorem[name=Question,numbered=no]{question*}

\declaretheoremstyle[bodyfont=\it,qed=$\lozenge$]{defstyle} 

\declaretheorem[numberlike=equation,style=defstyle]{definition}
\declaretheorem[unnumbered,name=Definition,style=defstyle]{definition*}

\declaretheorem[unnumbered,name=Example,style=defstyle]{example*}

\declaretheorem[unnumbered,name=Notation=defstyle]{notation*}

\declaretheorem[numberlike=equation,style=defstyle]{construction}
\declaretheorem[unnumbered,name=Construction,style=defstyle]{construction*}

\declaretheorem[unnumbered,name=Remark,style=defstyle]{remark*}

\usepackage{nth}
\usepackage{intcalc}
\usepackage{etoolbox}
\usepackage{xstring}
\hypersetup{
}

\usepackage{ifpdf}
\ifpdf
\else
\usepackage[quadpoints=false]{hypdvips}
\fi

\newcommand{\ECCCurl}[2]{http://eccc.weizmann.ac.il/report/\ifnumcomp{#1}{>}{93}{19}{20}#1/#2/}

\newcommand{\shortECCC}[2]{\texttt{\href{http://eccc.weizmann.ac.il/report/\ifnumcomp{#1}{>}{93}{19}{20}#1/#2/}{eccc:TR#1-#2}}}

\newcommand{\parseECCC}[1]{% Takes a string of the form TRxx/xxx or
\StrSubstitute{#1}{TR}{}[\tmpstring]%
\IfSubStr{\tmpstring}{/}{ %assuming string is of the form TRxx/xxx
\StrBefore{\tmpstring}{/}[\ecccyear]%
\StrBehind{\tmpstring}{/}[\ecccreport]%
}{% assuming string is of the form TRxx-xxx
\StrBefore{\tmpstring}{-}[\ecccyear]%
\StrBehind{\tmpstring}{-}[\ecccreport]%
}%
\shortECCC{\ecccyear}{\ecccreport}}

	\renewcommand{\vec}[1]{{\mathbf{#1}}}

	\makeatletter
	\newcommand{\va}{{\vec{a}}\@ifnextchar{^}{\!\:}{}}
	\newcommand{\vb}{{\vec{b}}\@ifnextchar{^}{\!\:}{}}
	\newcommand{\vc}{{\vec{c}}\@ifnextchar{^}{\!\:}{}}
	\newcommand{\vd}{{\vec{d}}\@ifnextchar{^}{\!\:}{}}
	\newcommand{\ve}{{\vec{e}}\@ifnextchar{^}{\!\:}{}}
	\newcommand{\vy}{{\vec{y}}\@ifnextchar{^}{\!\:}{}}
	\newcommand{\vs}{{\vec{s}}\@ifnextchar{^}{\!\:}{}}
	\newcommand{\vt}{{\vec{t}}\@ifnextchar{^}{\!\:}{}}
	\newcommand{\vx}{{\vec{x}}\@ifnextchar{^}{}{}}		%\vec{x} seems fine already
	\newcommand{\vz}{{\vec{z}}\@ifnextchar{^}{\!\:}{}}

	\newcommand{\vY}{{\vec{Y}}\@ifnextchar{^}{\!\:}{}}
	\newcommand{\vX}{{\vec{X}}\@ifnextchar{^}{}{}}		%\vec{x} seems fine already
	\newcommand{\vZ}{{\vec{Z}}\@ifnextchar{^}{\!\:}{}}
	\newcommand{\vG}{{\vec{G}}\@ifnextchar{^}{\!\:}{}}
	
	\makeatother

\newcommand{\cC}{{\mathcal{C}}}

\newcommand{\cH}{{\mathcal{H}}}

\newcommand{\F}{\mathbb{F}}

\newcommand{\set}[1]{\left\{#1\right\}}
\newcommand{\abs}[1]{\left|#1\right|}

\newcommand{\INW}{\mathrm{INW}}
\newcommand{\lin}{\mathrm{lin}}

\def\epsilon{\varepsilon} 
\let\eps\epsilon

\newcommand*\samethanks[1][\value{footnote}]{\footnotemark[#1]}

\date{}

\title{Pseudorandom Bits for Oblivious Branching Programs}
\author{
Rohit Gurjar\thanks{Department of Computer Science, Tel Aviv University, Tel Aviv, Israel, E-mails: \texttt{rohitgurjar0@gmail.com, benleevolk@gmail.com}.}
\and%
Ben Lee Volk\samethanks[1]
}
\begin{document}
\maketitle

\begin{abstract}
We construct a pseudorandom generator which fools read-$k$ oblivious branching programs and, more generally, any linear length oblivious branching program, assuming that the sequence according to which the bits are read is known in advance. For polynomial width branching programs, the seed lengths in our constructions are $\tilde{O}(n^{1-1/2^{k-1}})$ (for the read-$k$ case) and $O(n/ \log \log n)$ (for the linear length case). Previously, the best construction for these models required seed length $(1-\Omega(1))n$.
\end{abstract}

\thispagestyle{empty}
\newpage
\pagenumbering{arabic}

\section{Introduction}\label{sec:intro}

A \emph{Pseudorandom Generator} (PRG, for short), for a class of boolean functions $\cC$, is a family of efficiently computable functions $G_n : \set{0,1}^s \to \set{0,1}^n$ which fools functions in the class $\cC$, in the sense that for all $f : \set{0,1}^n \to \set{0,1}$ in $\cC$,
\[
\abs{ \Pr_{x \sim U_s} [ f(G_n(x)) = 1] - \Pr_{x \sim U_n} [ f(x) =1 ]} \le \eps,
\]
where $U_k$ denotes the uniform distribution on $\set{0,1}^k$, and $s=s(n,\eps)$ is called the \emph{seed length}.

A long line of work in complexity theory studies construction of PRGs for restricted classes of functions. One concrete motivation for these results is obtaining a ``black box'' derandomization of randomized algorithms from these restricted classes. More generally, this research program is aimed at shedding light on the power of randomness in computation in general, and on the limits of the use of randomness in algorithms with bounded resources.

One notable example in this area in Nisan's PRG for logarithmic space machines \cite{Nisan92}: Nisan constructed a PRG with seed length $O(\log^2 (n))$ which fools $\RL$ machines, that is, logarithmic space machines with read-once access to its randomness. More generally, Nisan's PRG also works in the non-uniform setting, and fools any function which is computed by a small width \emph{read once oblivious branching program} (see \autoref{sec:models} for a formal definition). 

In this more general setting, the seed length of Nisan's generator is $O(\log(n) \cdot (\log(n/\eps) + \log(w)))$, where $w$ denotes the \emph{width} of the branching program. Impagliazzo, Nisan and Wigderson \cite{INW94} gave a different construction with matching parameters, but to this day, and despite a large body of work on this topic, there is no better construction known for this model.

In the lack of better results, one possible avenue for improvement would be to obtain improved bounds in more restricted settings. One such challenge is to obtain an improved seed length in the bounded width case, i.e., when $w=O(1)$. Indeed, some progress was made in this setting, assuming more restrictive properties on the branching program (\cite{BRRY14, De11, KNP11, RSV13, Steinke12, SVW14}).

Another way to extend these results is to obtain PRGs against stronger models of computation. The saving in randomness in the works of Nisan \cite{Nisan92} and Impagliaazo, Nisan and Wigderson \cite{INW94} follows from the fact that in the execution of a read-once branching program on a specific input, each bit is accessed only once, and furthermore, the order of access is known in advance to the designer of the PRG. It is natural to ask to what extent these restrictions can be removed, en route to constructing PRGs against more general classes of computation. In \autoref{sec:related}, we review some of the progress made in this setting.

\subsection{Algebraic vs.\ Boolean Pseudorandomness}

We now make a small detour and review some relevant results from algebraic complexity. Polynomial Identity Testing (PIT) is the problem of deciding, given an algebraic computation device which computes a formal polynomial using the arithmetic operations $+$ and $\times$, whether it computes the zero polynomial. This problem admits an easy randomized algorithm which follows from the Schwartz-Zippel-DeMillo-Lipton Lemma \cite{Z79, S80, DL78}, and it is a major open problem to find an efficient deterministic algorithm, even for restricted classes of algebraic computation.

The algebraic analog of constructing PRGs is black-box identity testing: here, the goal is to construct a \emph{hitting set}, which is a small and efficiently constructible set $\cH$ such that for every non-zero polynomial $f$ in the class, there exists $\alpha \in \cH$ such that $f(\alpha) \neq 0$. It is not hard to show (see, e.g., \cite{sy}) that this is equivalent to constructing a \emph{generator}, which is a polynomial map $G : \F^s \to \F^n$ of small degree, such that for every non-zero $f$, $f \circ G$ is not the zero polynomial. The quality of the generator is measured by the seed length $s$ and the degree of the polynomial map $G$.

The algebraic analog of a read-once oblivious branching program is a model called read-once oblivious \emph{algebraic} branching programs (ROABPs). We omit the exact definition of this model from this informal introduction. Forbes and Shpilka \cite{FS13} obtained a hitting set of quasi-polynomial size for this model, or equivalently, a generator whose number of variables $s$ is $O(\log n)$, and whose degree is polynomial in the number of variables $n$, the width $w$ and the degree $d$ of the ROABP. Quantitatively, this is comparable to Nisan's generator,
and indeed, the intuition behind the Forbes-Shpilka generator is similar to Nisan's proof.

However, it interesting to note that both the challenges that were mentioned earlier in the context of boolean pseudorandomness have been met in the algebraic world: Forbes, Shpilka and Saptharishi \cite{FSS14} obtained a hitting set of quasi-polynomial size which works even when the order in which the variables are read is unknown. The construction was later improved by Agrawal, Gurjar, Korwar and Saxena \cite{AGKS15}, whose hitting set size matches the hitting set for the known order case.

In the bounded width setting, Gurjar, Korwar and Saxena \cite{GKS17} obtained a hitting set of polynomial size (over characteristic $0$, and assuming the order is known) by leveraging intuition for the INW generator \cite{INW94}.

It is thus interesting to see to what extent the progress in the algebraic world can help in obtaining improved PRGs for boolean computational devices.

\subsection{Results and Techniques}

In \cite{AFSSV16}, Anderson et al.\ obtained a subexponential time PIT algorithm for the model of read-$k$ oblivious algebraic branching program. Here, we adapt their techniques to the boolean analog of this model, and prove the following.

\begin{theorem}
\label{intro:thm:readk}
For every $k \ge 2$, there exists an efficiently computable function $G:\set{0,1}^s \to \set{0,1}^n$, where
\[
s = O \left( \exp(k^2) \cdot n^{1-1/2^{k-1}} \cdot \log(n) \cdot (\log(n/\eps) + k \log w) \right),
\]
which $\eps$-fools every function $f$ which is computable by a width-$w$ oblivious read-$k$ branching program, when the sequence according to which the variables are read in the branching program is known in advance.
\end{theorem}

The saving in randomness is more noticeable when $k$ is small. However, by exploiting the fact that the bound on $s$ remains sublinear even for slightly superconstant $k$, we can prove the following.

\begin{theorem}
\label{intro:thm:linear}
There exists an efficiently computable function $G:\set{0,1}^s \to \set{0,1}^n$, for $s=O(n / \log \log n)$, which $\eps$-fools every function $f$ which is computable by an oblivious branching program of length $O(n)$ and width $\poly(n)$, when the sequence according to which the variables are read in the branching program is known in advance.
\end{theorem}

Our techniques are mostly adaptation of the techniques used by \cite{AFSSV16} in the algebraic setting. To illustrate them, consider first a branching program that reads its variables twice in the order $x_1, x_2, \ldots, x_n, x_1, x_2, \ldots, x_n$. In the algebraic case, is it not very hard to show that such a width $w$ algebraic branching program can be simulated by a width $\poly(w)$ ROABP in the variable order $x_1, \ldots, x_n$, and this fact serves as the starting point of the construction in \cite{AFSSV16}. This fact, however, is no longer true for boolean branching programs. As an example, consider the ``address function'' which receives as an input $y \in \set{0,1}^n$ and $z \in \set{0,1}^{\log n}$, interprets $z$ as an integer in $[n]$ and outputs $y_z$. In the variable order $y_1, y_2, \ldots, y_n, z_1, \ldots, z_{\log n}$, this function requires exponential width, since the branching program essentially has to remember all $y$ bits
before it sees $z$ bits. But if the branching program is allowed to read the input twice in this order, polynomial width suffices.

Fortunately, it turns out that the generator by Impagliazzo, Nisan and Wigderon \cite{INW94} also fools branching programs that read their input twice in the same order. This is essentially because this generator is useful against any model which can be simulated by a ``low communication'' protocol on a ``simple'' network topology (see \autoref{sec:INW} for further discussion). Thus, this model already has a PRG with seed length $\polylog(n)$ (assuming $w=\poly(n)$).

Generalizing a bit further, we can consider branching programs that read their input in the order $x_1, x_2, \ldots, x_n, x_{\pi(1)}, \ldots, x_{\pi(n)}$ for some arbitrary permutation $\pi$. Here, the basic idea in \cite{AFSSV16} was to argue, using the Erd\H{o}s-Szekeres Theorem, that the sequence $x_{\pi(1)}, \ldots, x_{\pi(n)}$ must contain either a monotonically increasing sequence of length $\sqrt{n}$, or a monotonically decreasing sequence of the same length. Assuming that the sequence is increasing (the decreasing case is handled similarly), we obtain a set of $\sqrt{n}$ variables $y_1, \ldots, y_{\sqrt{n}}$ such that the branching program, restricted to only these variables, is exactly of the form required by the previous argument.

We continue inductively to find a (slightly shorter) monotone sequence in the remaining variables. This process can be shown to terminate after $O(\sqrt{n})$ applications of the Erd\H{o}s-Szekeres Theorem,  and the final generator is obtained by applying the INW generator with an independent seed to each of the $O(\sqrt{n})$ sets obtained in this process, for a total seed length of $O(\sqrt{n} \cdot \polylog(n))$. 

Similarly, one can consider read-$k$ branching programs whose reading order is 
\begin{equation}
\label{eq:kpass}
x_1, \ldots, x_n, x_{\pi_1(i)}, \ldots, x_{\pi_1(n)}, \ldots, x_{\pi_{k-1}(1)}, \ldots, x_{\pi_{k-1}}(n),
\end{equation}
for $k-1$ permutations $\pi_1, \ldots, \pi_{k-1}$. By iteratively applying the  Erd\H{o}s-Szekeres Theorem, we can find a sequence of length $n^{1/2^{k-1}}$ which is monotone in each of the $k$ reads, argue as before that the branching program restricted to these variables is fooled by the INW generator, and continue the argument as before. The iterative application of the Erd\H{o}s-Szekeres argument accounts for most of the loss in the parameters, but it is unfortunately unavoidable in this approach (see the discussion in \cite{AFSSV16}).

More generally, we want to handle any sequence in which every variable appears at most $k$ times, even if it is not of the form \eqref{eq:kpass}. To do this, Anderson et al.\ \cite{AFSSV16} defined the notion of a ``$k$-regularly-interleaving sequence'', which we define in \autoref{sec:readkseq}, and enables us to similarly partition the set of variables $X$ into $t$ disjoint sets $Y_1, \ldots, Y_t$, such that the INW fools every branching programs in the variables of $Y_i$ under any restriction of the variables in all other sets, while $t$ remains sublinear in $n$.

\subsection{Related Work}
\label{sec:related}

There are several works that consider the problem of constructing pseudorandom distributions for related models. Here we review some of the related results.

Impagliazzo, Meka and Zuckerman \cite{IMZ12} constructed a very general PRG that fools \emph{every} branching program with $s$ vertices, with seed length $s^{1/2 + o(1)}$. For the case of branching programs of length $O(n)$, the size is $O(w \cdot n)$, and thus this is meaningful only when $w=o(n)$, whereas our result remains non-trivial for any polynomial, and even super-polynomial, width.

Bogdanov, Papakonstantinou and Wan \cite{BPW11, BPW12} constructed explicit PRGs for read-once branching programs and more generally for oblivious branching programs of linear length. In their construction, the seed length is $(1-\Omega(1))n$, whereas our seed length is sublinear in $n$. However, an advantage of their construction is that it works even when the reading order of the bits is unknown, while we require it to be known in advance. Haramaty, Lee and Viola obtained some improved bounds for the related model of \emph{product tests} \cite{HLV17}.

Finally, we discuss the pseudorandom generator of Impagliazzo, Nisan and Wigderson \cite{INW94}, which is also a useful tool in our construction. This work is often cited in the context of derandomizing logarithmic space or read-once oblivious branching programs, but is in fact applicable in other contexts which can be modelled as a network of processors, each flipping its own random bit. The parameters of the generator depend on the ``simplicity'' of the network graph, and the total communication between the processors.

In the common setting of read-once oblivious branching programs, the network consists of $n$ processors, where the $i$-th processor 
reads the random bit $x_i$, and the network graph is a simple path. This graph is simple in the technical sense required by \cite{INW94}, and by the read-once property, the computation of a width $w$ branching program can be simulated by each processor receiving at most one message and sending at most one message, each of length at most $\log(w)$.
The message basically encodes the index of a node in the next layer of the branching program. 

Considering read-twice branching programs, it is clear that the communication remains bounded if, for example, one considers branching programs which read their input twice in the same order (i.e., $x_1, x_2, \ldots, x_n, x_1, x_2, \ldots, x_n$). However, the situation changes when one considers general read-twice sequences: in this setting, it is always possible to bound the communication by allowing a more complicated network structures, e.g., a clique between the $n$ processors which will allow each pair to communicate directly. However, in this case the network structure is no longer ``simple'' in the sense required by \cite{INW94}.

We insist on the network being a 
path, in which case it is convenient to model the variable access of the branching program as a Turing machine head, which can move at each step left or right (but cannot ``jump'' many cells in one step). This approach is also taken by \cite{INW94}, which show that their construction works as long as one can bound the number of times the Turing machine head visit each cell (see \autoref{thm:INW-gen} for a formal statement). Their seed length is proportional to the maximum number of times a cell is visited.
	
To understand the subtleties, it is useful to consider the following read-twice sequence:
\[
x_1, x_2, \ldots, x_n, x_1, x_n, x_2, x_{n-1}, \ldots, x_{n/2}, x_{n/2+1}.
\]
If the order on the Turing machine tape is $(x_1, x_2, \ldots, x_n)$ then 
the Turing machine head reading the sequence would have to visit the $(n/2)$-th cell $\Omega(n)$ times, even though each element only appears twice in the sequence.

In this example, if the designer of the PRG is allowed to look at the sequence of bits --- which is the case in our setting --- they can ``imagine'' that the order of the bits on the tape is   $x_1, x_n, x_2, x_{n-1}, \ldots, x_{n/2}, x_{n/2+1}$, 
%which 
so that the above sequence
can be handled by a Turing machine head which reads each cell 3 times, and then instantiate the generator with this order.

Of course, in a more general case there is no guarantee that we can fix a order on the tape which ensures each cell is visited only a small number of times.
In fact, there exists a read-twice sequence $x_1, x_2, \ldots, x_n, x_{\pi(1)}, \ldots, x_{\pi(n)}$ %,  x_{\sigma(1)}, \ldots, x_{\sigma(n)}$, 
for some permutation $\pi$, such that for any fixed order on the tape, there will be a cell which would be visited $\Omega(n)$ times. 
Thus, one cannot hope to directly apply the INW generator. 
To overcome this, we follow Anderson et al.\ \cite{AFSSV16} and partition the set of variables $X$ into $t$ disjoint sets $Y_1, \ldots, Y_t$ (\autoref{thm:good-subsequences}) using Erd\H{o}s-Szekeres Theorem. 
The partition has the property that the given sequence restricted to any part $Y_i$
can be traversed by Turing machine head while visiting each cell a bounded number of times (\autoref{lem:good-for-INW}). 
Thus, it would suffice to plug in independent copies INW generator to each set $Y_i$ (\autoref{lem:hybrid}).

\section{Preliminaries}
\label{sec:prelim}

\subsection{Notation}
Let $[n] = \set{1, 2, \ldots, n}$. For a partition of $[n] = A_1 \sqcup A_2 \sqcup \cdots \sqcup A_r$ with $|A_i| = m_i$, and distributions $D_i$ on $\set{0,1}^{m_i}$, we denote by
\[
D_1^{A_1} \times D_2^{A_2} \times \cdots \times D_r^{A_r}
\]
the distribution on $\set{0,1}^n$ obtained by sampling \emph{independently} a vector $\vb_i \in \set{0,1}^{m_i}$ from $D_i$, for $i \in [r]$, and obtaining a string $\vb \in \set{0,1}^n$ by plugging $\vb_i$ in the coordinates indexed by $A_i$.

This notation is also used for functions in a natural way: if $G_i : \set{0,1}^{s_i} \to \set{0,1}^{m_i}$ is any function, $G_1^{A_1} \times \cdots \times G_r^{A_r}$ is a function from $\set{0,1}^{s_1 + \cdots + s_r}$ to $\set{0,1}^n$ obtained by applying $G_1$ to the first $s_1$ input bits and plugging the result in the coordinates of $A_1$, and so on.

We often consider models which compute boolean functions over a variable sets $X = \set{x_1, \ldots, x_n}$ and various partitions of the sets of variables $Y_1 \sqcup Y_2 \sqcup \cdots \sqcup Y_r$. By considering the indices of the variables in each set, such a partition corresponds naturally to a partition of $[n]$ and thus we use similar notations as above with the $Y_i$'s in the superscript.

\subsection{Computational Models}
\label{sec:models}
A branching program $B$ on a variable set $X=\set{x_1, \ldots, x_n}$ is a directed acyclic graph, with a unique source vertex, two sink vertices labeled ``accept'' and ``reject'', and where every non-sink vertex is labeled by one of the $n$ variables and has exactly two outgoing edges, labeled $0$ and $1$. $B$ naturally define a boolean function $B : \set{0,1}^n \to \set{0,1}$ by considering the path an input $x$ induces in the graph. In our case, the branching programs will be \emph{layered}, that is, the vertex set can be partitioned into $m$ layers, with every edge going from layer $i-1$ to $i$.

Such a branching program is said to be \emph{oblivious} if on every layer, all the vertices are labelled by the same variable, namely, the program reads its input in a fixed order which is independent in the value it has read so far. An oblivious branching program is said to be also \emph{read-once} if every variable appears as a label in at most one layer, and more generally \emph{read-$k$} if every variable appears in at most $k$ layers.\footnote{The modifiers read-once and read-$k$ can be used, and have been used, also in the context of non-oblivious branching programs. %% add citations
In the more general context, one has to distinguish between a syntactic definition and a semantic definition. As the two definitions coincide in the oblivious case, which is the only case we consider, we omit this discussion.
} Without loss of generality, we may assume each variable is read exactly $k$ times.

\subsection{Read-$k$ sequences}
\label{sec:readkseq}

Let $X = \set{x_1, \ldots, x_n}$, and $S \in X^m$ be a sequence of $m$ elements from $X$. $S$ is said to be a read-$k$ sequence over $X$ if every element $x \in X$ appears exactly $k$ times (and in that case $m=nk$). For a set $Y \subseteq X$ we let $S|_{Y}$ denote the subsequence of $S$ which is obtained by keeping only the elements in $Y$, and erasing all other elements. For $i \in [k]$, we denote by $S^{(i)}$ the subsequence of $S$ which consists of the $i$-th occurrences of the elements. In other words, $S^{(i)} \in X^n$ is a permutation of $X$ according to the order of their $i$-th occurrences. Without loss of generality and by renaming variables, if necessary, we always assume $S^{(1)} = (x_1, \ldots, x_n)$ is the identity permutation. Similarly, for $i \neq j \in [k]$, we use the notation $S^{(i,j)}$ for the subsequence of $S$ which consists of the $i$-th and $j$-th occurrences of all elements. We also associate a natural linear order on $X$ by letting $x_1 < x_2 < \cdots < x_n$.

We now cite relevant definitions from \cite{AFSSV16}. We begin with the definition of a per-read-monotone sequence.

\begin{definition}
A read-$k$ sequence is said to be \emph{per-read-monotone} if for all $i \in [k]$, $S^{(i)}$ is either monotonically increasing or monotonically decreasing.
\end{definition}
Similarly, a read-$k$ sequence is said to be \emph{per-read-increasing (or per-read-decreasing)} if for all $i \in [k]$, $S^{(i)}$ is monotonically increasing (or decreasing).

\begin{definition}
A read-$2$ sequence $S$ over $X=\set{x_1,\ldots,x_n}$ is said to be \emph{2-regularly-interleaving} if there exists a partition $X=X_1 \sqcup X_2 \sqcup \cdots \sqcup X_t$, such that for every $i \in [t]$, the following two conditions hold:
\begin{enumerate}
\item The sequence $S$ can be partition into $t$ read-2 sequences $\set{S_i}_{i \in [t]}$ such that $S_i \in X_i^{2|X_i|}$, and $S=(S_1, \ldots, S_t)$ is the concatenation of $S_1, \ldots, S_t$.
\item Each $S_i$ as above can be partitioned into two subsequences $S_{i,1}$ and $S_{i,2}$, such that for $c \in \set{1,2}$, $S_{i,c}$ contains the $c$-th occurrences of $X_i$, and $S_i$ equals the concatenation of $S_{i,1}$ and $S_{i,2}$.
\end{enumerate}

A read-$k$ sequence is \emph{$k$-regularly-interleaving} if for all $i \neq j \in [k]$, the subsequence $S^{(i,j)}$ is 2-regularly-interleaving.
\end{definition}

The following theorem was proved in \cite{AFSSV16}. Roughly, it says that for small $k$, every read-$k$ sequence can be partitioned into a sublinear number of subsequences, each of which is per-read-monotone and $k$-regularly-interleaving.

\begin{theorem}[\cite{AFSSV16}]
\label{thm:good-subsequences}
Let $S$ be a read-$k$ sequences over $X=\set{x_1, \ldots, x_n}$. Then, $X$ can be partitioned into $t$ disjoint subsets $Y_1 \sqcup Y_2 \sqcup \cdots \sqcup Y_t$, such that 
\begin{enumerate}
\item The subsequence $S_i=S|_{Y_i}$ is per-read-monotone and $k$-regularly interleaving.
\item $t \le \exp(k^2) \cdot n^{1-1/2^{k-1}}$.
\end{enumerate}
Further, this partition can be computed, given $S$, in time $\poly(k,n)$.
\end{theorem}

The following observation about the structure of per-read-monotone sequences was also made in \cite{AFSSV16}. Intuitively, it says that in a per-read-monotone sequence, increasing and decreasing subsequences cannot intersect. That is, a per-read-monotone sequence can be partitioned into subsequences, which are alternately per-read-increasing and per-read-decreasing.

\begin{proposition}
\label{prop:increasing-decreasing}
Let $S$ be a read-$k$ per-read-monotone sequence over $X=\set{x_1, \ldots, x_n}$. Then $S$ is a concatenation of $t \le k$ subsequences $S=(T_1, \ldots, T_t)$ such that:
\begin{enumerate}
\item $t \le k$.
\item There exist $1=i_1 < i_2 < i_3 < \ldots < i_{t-1} < i_t \le k$ such that for all $i_j \le c < i_{j+1}$, $S^{(c)}$ is contained in $T_j$.
\item For all odd $j$ (even, respectively), all the subsequences $S^{(c)}$ that appear in $T_j$ are monotonically increasing (decreasing, respectively).
\end{enumerate}
\end{proposition}

\subsection{The Impagliazzo-Nisan-Wigderson Generator}
\label{sec:INW}

As mentioned in \autoref{sec:related}, we use the INW generator as an important tool in our construction. Here we cite their specific construction which we use, which can handle multiple reads of the input, as long as the ``total communication'' is bounded.

\begin{theorem}[\cite{INW94}, Theorem 3]
\label{thm:INW-gen}
There exists a generator $G^{\INW}_{n,d,\eps}:\set{0,1}^s \to \set{0,1}^n$ which $\eps$-fools every width $w$ oblivious branching program, in which the reading order can be simulated by a Turing machine head which visits every cell at most $d$ times. The seed length $s$ is
$O(\log n \cdot (d \log w + \log(n/\eps)))$.
\end{theorem}

\subsection{Combining Generators}
For an oblivious branching program $B$ over the variable set $X$, a subset $Z \subseteq X$, and a 
bit vector $\vb \in \set{0,1}^{\abs{Z}}$, let $B|_{Z=\vb}$ denote 
the branching program obtained by fixing the variables in $Z$ according to the values given by $\vb$.
Observe that for all $\vb$, $B|_{Z=\vb}$ is an oblivious branching program over the variable set $X \setminus Z$.
\begin{lemma}
\label{lem:hybrid}
Let $B$ be an oblivious branching program of width $w$ over the variable set $X=\set{x_1, \ldots, x_n}$.
Let $Y \subseteq X$ be such that $|Y|=m$ and let $Z := X \setminus Y$. 
Let $D$ be a distribution on $\{0,1\}^m$ that $\eps$-fools $B|_{Z=\vb}$, for all $\vb \in \{0,1\}^{n-m}$,
%$m$-variate width $w$ oblivious branching programs, 
and let $D'$ be \emph{any} distribution on $\{0,1\}^{n-m}$.
Denote $\mu_1 = U_m^Y \times D'^{Z}$ and $\mu_2 = D^Y \times D'^{Z}$. Then, it holds that
\[
\abs{\Pr_{x \sim \mu_1} [ B(x)=1] - \Pr_{x \sim \mu_2} [ B(x)=1]} \le \eps
\]
\end{lemma}
\begin{proof}
From the lemma hypothesis,
\begin{equation}
\label{eq:Yfools}
\abs{\Pr_{y \sim U_m}[B|_{Z=\vb}(y)=1] - \Pr_{y \sim D}[B|_{Z=\vb}(y)=1]} \le \eps,
\end{equation}

%To complete the argument, 
We observe that the distribution of $B|_{Z=\vb}(y)$ where $y$ is chosen according to $U_m$ (or $D$, respectively) is the same as the marginal distribution of $B(x)$ conditioned on $Z=\vb$, where $x$ is chosen from $\mu_1$ (or $\mu_2$, respectively).

Under these notations,
\[
\Pr_{x \sim \mu_1} [ B(x)=1] = \sum_{\vb} \Pr_{x \sim U_m}[B(x)=1 | Z=\vb] \cdot \Pr[Z=\vb]
=  \sum_{\vb} \Pr_{y \sim U_m}[B|_{Z=\vb}(y)=1] \cdot \Pr[Z=\vb],
\]
and similarly,
\[
\Pr_{x \sim \mu_2} [ B(x)=1] = \sum_{\vb} \Pr_{y \sim D}[B|_{Z=\vb}(y)=1] \cdot \Pr[Z=\vb].
\]
Thus, using \eqref{eq:Yfools} it follows that
\begin{align*}
\abs{\Pr_{x \sim \mu_1} [B(x)=1] - \Pr_{x \sim \mu_2} [ B(x)=1]} & = 
\abs{  \sum_{\vb} \left( \Pr_{y \sim U_m}[B|_{Z=\vb}(y)=1] - \Pr_{y \sim D}[B|_{Z=\vb}(y)=1] \right)\cdot \Pr[Z=\vb]} \\ &
\le
\sum_{\vb} \abs{ \Pr_{y \sim U_m}[B|_{Z=\vb}(y)=1] - \Pr_{y \sim D}[B|_{Z=\vb}(y)=1] } \cdot  \Pr[Z=\vb]  \\
& \le \eps. \qedhere
\end{align*}
\end{proof}

\section{Pseudorandom Generator for Read-$k$ Oblivious Branching Programs}
\label{sec:readk}

We begin by showing that the generator $G^{\INW}$ from \autoref{thm:INW-gen} is pseudorandom against read-$k$ oblivious branching programs that read their inputs in a per-read-monotone and $k$-regularly-interleaving fashion. To that end, we show that such sequences satisfy the properties required by the theorem.

Recall from \autoref{prop:increasing-decreasing}, that any read-$k$, per-read-monotone sequence is a concatanation of 
subsequences which are alternately per-read-increasing and per-read-decreasing. 
The following Lemma from \cite{AFSSV16} says that in a per-read-increasing and $k$-regularly interleaving sequence there are no upward jumps (the case of per-read-decreasing is analogous). 

\begin{lemma}
\label{lem:reg-interleaving-structure}
Let $S$ be a read-$k$, per-read-increasing, and $k$-regularly-interleaving sequence over $X=\set{x_1, \ldots, x_n}$.
 Let $\ell \in [kn]$ be an integer, and suppose that $x_i$ appears in the $\ell$-th position in $S$, and for some $j > i$, $x_j$ appears in position $\ell+1$. Then $j=i+1$.
\end{lemma}

We use this lemma to show that a per-read-monotone and $k$-regularly-interleaving sequence satisfies the properties required by \autoref{thm:INW-gen}.

\begin{lemma}
\label{lem:good-for-INW}
Let $S$ be a read-$k$, per-read-monotone, and $k$-regularly-interleaving sequence over $X=\set{x_1, \ldots, x_n}$. Then, when modelling the variable access of $S$ as a Turing machine head, the head visits every cell at most $2k$ times.
\end{lemma}

\begin{proof}
We first argue that it is enough to consider the case when we the given sequence is per-read-increasing.
From \autoref{prop:increasing-decreasing}, 
it follows that $S$ is concatenation of $t$ subsequences $S=(T_1,\dots,T_t)$ such that 
for each $r \in [t]$, $T_r$ is a read-$k_r$ and $k_r$-regularly-interleaving sequence over $X$
for some $k_1,\dots,k_t$ with $k_1+k_2+\cdots + k_t = k$.
Moreover, for all odd $r$ (even, respectively), $T_r$ is per-read-increasing (decreasing, respectively). 
In particular, this means that for all odd $r$, $T_r$ starts with $x_1$ and ends with $x_n$.
On the other hand for all even $r$, $T_r$ starts with $x_n$ and ends with $x_1$.
Thus, when moving from $T_r$ to $T_{r+1}$ the Turing machine head does not visit any new cell.
We claim that while traversing $T_r$, the Turing machine head visits every cell at most $2 k_r$ times. 
This would imply that while traversing $S$, the head visits every cell at most $\sum_{r=1}^t 2 k_r = 2k $ times. 

Now, consider the sequence $T_r$, which is a read-$k_r$, per-read-increasing (the decreasing case is similar), and $k_r$-regularly-interleaving sequence.
Obviously, the head needs to visit the $i$-th cell whenever $x_i$ appears in the sequence. This can happen, by assumption, at most $k_r$ times. However, it also needs to pass through $x_i$ whenever, for $j<i<h$, $x_j$ appears in the sequence and followed by $x_h$, or $x_h$ is followed by $x_j$, and thus our goal is bound the number of times these can happen.

We claim the first transition cannot happen at all in $S$. For suppose $x_j$ appears at position $\ell$, and is immediately followed by $x_h$ in position $\ell+1$. Since $h>i>j$, this contradicts \autoref{lem:reg-interleaving-structure}.

As for the second type of transition, we claim there can be at most $k_r$ of these. 
By \autoref{lem:reg-interleaving-structure}, for any  $h'>i>j'$ and for any appearance of $x_{h'}$ after $x_{j'}$ in $T_r$,
the element $x_i$ must appear in between them. 
Since $x_i$ can appear in $T_r$ at most $k_r$ times, this establishes the claim.
%Suppose $x_h$ appears in position $\ell$, and $x_j$ in position $\ell+1$. Since $j<i$, by \autoref{lem:reg-interleaving-structure}, the head cannot cross the $i$-th cell without reading $x_i$, namely, each following appearance of $x_{k'}$, for $k'>i>j$ has to be preceded by $x_i$. Since $x_i$ appears at most $k$ times, this establishes the claim.
\end{proof}

We are now ready to present the construction of our pseudorandom generator for read-$k$ oblivious branching programs.

\begin{construction}
\label{con:generator-readk}
Let $S$ be a read-$k$ sequence over $X=\set{x_1, \ldots x_n}$, and let $Y_1,\ldots,Y_t$ be as promised by \autoref{thm:good-subsequences}. Let $n_i = |Y_i|$, and $s_i$ be the seed length of $G^{\INW}_{n_i,2k,\eps/n}(\cdot)$ as given by \autoref{thm:INW-gen}, 
%such that,  -> that is
that is,
$s_i = O(\log(n) \cdot (\log(n/\eps) + k \log w))$, and let $s = \sum_{i=1}^t s_i$.
Define $G^{k}_{\eps} : \set{0,1}^s \to \set{0,1}^n$, by
\[
G^{k}_{\eps}(\vy) = \left( G^{\INW}_{n_1, 2k, \eps/n}(\vy_1) \right) ^{Y_1} \times \left( G^{\INW}_{n_2, 2k, \eps/n}(\vy_2) \right) ^{Y_2} \times \cdots \times \left( G^{\INW}_{n_t, 2k, \eps/n}(\vy_t)\right)^{Y_t},
\]
where $\vy=(\vy_1, \ldots, \vy_t)$ and $\vy_i \in \set{0,1}^{s_i}$. 
\end{construction}

\begin{theorem}
\label{thm:generator-readk}
Let $S$ be a read-$k$ sequence. The generator $G^{k}_{\eps}$ from \autoref{con:generator-readk} $\eps$-fools every read-$k$ oblivious branching program which reads the variables in the order prescribed by $S$. The seed length $s$ is
\[
O \left( t \cdot \log(n) \cdot (\log(n/\eps) + k \log w) \right) = O \left( \exp(k^2) \cdot n^{1-1/2^{k-1}} \cdot \log(n) \cdot (\log(n/\eps) + k \log w) \right).
\]
\end{theorem}

\begin{proof}
The bound of the seed length follows directly from the construction.

The proof that it indeed $\eps$-fools read-$k$ oblivious branching program is by a standard hybrid argument, using \autoref{lem:hybrid}.

Let $B$ be any branching program as stated in the theorem, and let $Y_1, \ldots, Y_t$ be the partition as described in \autoref{con:generator-readk}. Recall that $|Y_i|=n_i$. Denote by $U_{n_i}$ the uniform distribution on $\set{0,1}^{n_i}$, and by $D_i$ the distribution of $G^{\INW}_{n_i, 2k, \eps/n}(\vy_i)$, with $\vy_i$ randomly and uniformly picked from $\set{0,1}^{s_i}$.

Now observe that the distribution of a randomly and uniformly 
seeded
 $G^{k}_{\eps}$ is given by $\mu_t := D_1^{Y_1} \times D_2^{Y_2} \times \cdots \times D_t^{Y_t}$, whereas $\mu_0 := U_n = U_{n_1}^{Y_1} \times U_{n_2}^{Y_2} \times \cdots \times U_{n_t}^{Y_{n_t}}$. Similarly, for every $0 \le j \le t$, define
\[
\mu_j = D_{1}^{Y_1} \times \cdots \times D_{j}^{Y_{j}} \times U_{n_{j+1}}^{Y_{j+1}} \times \cdots \times U_{n_t}^{Y_t}.
\]
Consider any $1 \le j \le t$.
Let $Z_{j} = X \setminus Y_{j}$.
Recall that for any bit vector $\vb \in \{0,1\}^{n-n_{j}}$, the restriction $B|_{Z_{j}=\vb}$ is a width $w$ oblivious branching program 
over variables $Y_{j}$. 
From \autoref{thm:good-subsequences}, the sequence $S|_{Y_j}$ is a read-$k$, per-read-monotone, $k$-regularly interleaving sequence. 
Thus, by \autoref{lem:good-for-INW} and \autoref{thm:INW-gen}, the distribution $D_j = G^{\INW}_{n_j, 2k, \eps/n}(\vy_j)$ $\varepsilon$-fools 
the branching program $B|_{Z_{j}=\vb}$.
Now, we apply \autoref{lem:hybrid} on $B$ with $D$ as $D_j$ and $D'$ as
$D_{1}^{Y_1} \times \cdots \times D_{j-1}^{Y_{j-1}} \times U_{n_{j+1}}^{Y_{j+1}} \times \cdots \times U_{n_t}^{Y_t}$.
We get that for each $1 \le j \le t$,
\[
\abs{\Pr_{\vx \sim \mu_{j-1}} [B(\vx)=1] - \Pr_{\vx \sim \mu_{j}} [ B(\vx) = 1]} \le \varepsilon/n,
\]
which, by the triangle inequality, implies that
\[
\abs{\Pr_{\vx \sim \mu_0} [B(\vx)=1] - \Pr_{\vx \sim \mu_{t}} [ B(\vx) = 1]} \le t \cdot \varepsilon/n \le \varepsilon,
% t+1 -> t
\]
as in the statement of the theorem. 
\end{proof}

\section{Pseudorandom Generator for Linear Length Oblivious Branching Programs}
\label{sec:linear-length}

Our generator for general linear length oblivious branching programs is based on the simple observation that in the generator from \autoref{sec:readk}, the seed length remains sublinear even for $k=k(n)$ which is slightly super-constant, whereas if the length of an oblivious branching program is at most $cn$, the number of variables which appear more than $k$ times is at most $\frac{c}{k}n$, which is sublinear. Thus, these variables can be just sampled uniformly.

\begin{construction}
\label{con:generator-linlength}
Let $S \in X^{cn}$ be a sequence of length $cn$ over $X=\set{x_1, \ldots x_n}$, and set  $k = k(n)= (\log \log n)/2$. A variable is said to be \emph{frequent} if it appears more than $k$ times in $S$. Let $F$ be the set of frequent variables, so we know that $|F| \le cn/k$.
Let $s_1$ be the seed length of $G^{k}_{\varepsilon}$ from \autoref{con:generator-readk}, and $s_2 = |F|$.

Define $G^{\lin} : \set{0,1}^s \to \set{0,1}^n$, by
\[
G^{\lin}(\vy) = \left( G^{k}_{\varepsilon}(\vy_1) \right) ^{X\setminus F} \times \left( \vphantom{G^{k}_{\varepsilon/2}} \vy_2 \right) ^{F} 
\]
where $\vy=(\vy_1, \vy_2)$ and $\vy_i \in \set{0,1}^{s_i}$. 
\end{construction}

\begin{theorem}
\label{thm:oblivious-lin}
Let $S \in X^{cn}$ be a sequence over $X=\set{x_1, \ldots, x_n}$ of length $cn$. The generator  $G^{\lin} : \set{0,1}^s : \set{0,1}^n$ from \autoref{con:generator-linlength} $\eps$-fools every oblivious  branching program $B$ of width $w$ that reads its variables in the order prescribed by $S$. The seed length $s$ is $O(\frac{n}{\log \log n})$ for $w=\poly(n)$.
\end{theorem}

\begin{proof}
Since $S|_{X\setminus F}$ is a read-$k$ sequence,
from \autoref{thm:generator-readk}, the generator $G^{k}_{\varepsilon}$ from \autoref{con:generator-readk} $\varepsilon$-fools the branching program $B|_{F=\vb}$ for any $\vb \in \{0,1\}^{\abs{F}}$.	
Thus, from \autoref{lem:hybrid}, the generator $G^{\lin}$ from \autoref{con:generator-linlength} $\eps$-fools $B$.
%such branching program is a consequence of \autoref{thm:generator-readk} and \autoref{lem:hybrid}. 
The bound on the seed length follows from the seed length of \autoref{con:generator-readk}.
\end{proof}

\section*{Acknowledgment}
We thank Andrej Bogdanov for useful comments on an earlier version of this text.

\bibliographystyle{customurlbst/alphaurlpp}
\bibliography{references}

\end{document}